\documentclass{article}
\usepackage{ijcai26}

\usepackage{times}
\usepackage{soul}
\usepackage{url}
\usepackage[hidelinks]{hyperref}
\usepackage[utf8]{inputenc}
\usepackage[small]{caption}
\usepackage{graphicx}
\usepackage{amsmath}
\usepackage{amsthm}
\usepackage{booktabs}
\usepackage{algorithm}
\usepackage{algpseudocode}
\usepackage[switch]{lineno}

\usepackage{thm-restate}
\usepackage{pifont}
\usepackage{multirow}
\usepackage{xcolor}

\usepackage{subfigure}
\usepackage{diagbox}
\usepackage{dblfloatfix}

\usepackage{amssymb}

\usepackage{natbib}
\usepackage{url}

\newtheorem{definition}{Definition}

\newcommand{\sed}[1]{\underline{\textit{#1}}}

\newcommand{\ut}{\boldsymbol{u}}

\newcommand{\R}{\mathbb{R}}

\newcommand{\cC}{\mathcal{C}}

\newcommand{\cL}{\mathcal{L}}

\newcommand{\cS}{\mathcal{S}}

\title{Efficient Approximation Algorithms for Fair Influence Maximization under Maximin Constraint}

\author{
Xiaobin Rui$^{1}$
\and
Qiangpeng Fang$^1$\and
Chen Peng$^1$\and
Jilong Shi$^1$\and
Zhixiao Wang$^{1,*}$\And
Wei Chen$^{2,*}$\\
\affiliations
$^1$China University of Mining and Technology, Xuzhou 221116, China\\
$^2$Theory Center of Microsoft Research Asia, Beijing 100080, China\\
\emails
\{ruixiaobin, fangqp, 08222452, jlshi, zhxwang\}@cumt.edu.cn,
weic@microsoft.com,
}

\begin{document}

\maketitle

\begin{abstract}
Fair Influence Maximization (FIM) seeks to mitigate disparities in influence across different groups and has recently garnered increasing attention. 
A widely adopted notion of fairness in FIM is the maximin constraint, which directly requires maximizing the utility (influenced ratio within a group) of the worst-off group. 
Despite its intuitive formulation, designing efficient algorithms with strong theoretical guarantees remains challenging, as the maximin objective does not satisfy submodularity, a key property for designing approximate algorithms in traditional influence maximization settings.
In this paper, we address this challenge by proposing a two-step optimization framework consisting of Inner-group Maximization (IGM) and Across-group Maximization (AGM). 
We first prove that the influence spread within any individual group remains submodular, enabling effective optimization within groups.
Based on this, IGM applies a greedy approach to pick high-quality seeds for each group.
In the second step, AGM coordinates seed selection across groups by introducing two strategies: Uniform Selection (US) and Greedy Selection (GS). 
We prove that AGM-GS holds a $(1-1/e-\varepsilon)$ approximation to the optimal solution when groups are completely disconnected, while AGM-US guarantees a roughly $\frac{1}{m}(1-1/e-\varepsilon)$ lower bound regardless of the group structure, with $m$ denoting the number of groups.
\end{abstract}

\section{Introduction}

Influence maximization (IM)~\citep{IMMsurvey1,IMMsurvey2}, a fundamental problem in social network analysis, aims to select a small set of seed nodes that maximizes the spread of influence in a network. 
While classical approaches primarily focus on maximizing total influence, recent work~\citep{groupfair,gaps,unifying,rui2023,optimaltransport,cff} has addressed the issue of fairness in the IM problem, especially in scenarios where the population is divided into distinct demographic groups.
In such a scenario, traditional IM methods may disproportionately benefit well-connected groups, leaving underrepresented groups with significantly lower influence coverage.

To depict the fairness in the IM problem, existing studies have come up with various notions of fairness, including maximin fairness~\citep{groupfair,gaps}, equity fairness~\citep{unifying}, welfare fairness~\citep{welfare,rui2023}, concave fairness~\citep{timecritical,cff}, {\em etc}.
Among these notions, maximin fairness, also known as the maximin constraint, stands out for its simplicity and intuitive appeal. 
It imposes a clear requirement: the utility ({\em i.e.}, the proportion of influenced individuals) of the worst-off group should be maximized. 
In other words, the algorithm should prioritize improving the influence spread for the most disadvantaged group, thereby ensuring a more balanced and fairer outcome.


To tackle the FIM problem under the maximin constraint, \cite{groupfair} applied a Frank-Wolfe style multiobjective optimization approach, treating the utility of each group as a separate objective.
However, their approach is computationally inefficient and was only evaluated on small networks with fewer than a hundred nodes in their original paper.
\cite{gaps} proposed a novel $\phi$-MEAN fairness, where setting $\phi = -\infty$ recovers the maximin objective.
They outlined a naive Greedy algorithm that, in each iteration, picks the node that maximizes the objective function as a seed.
The naive Greedy achieves better objective values but at the cost of high computational complexity.
\cite{milp} reformulated the maximin objective into an approximate setting and addressed it using Mixed-Integer Linear Programming (MILP). 
While the approach is theoretically sound, it suffers from scalability issues and is limited in practice to networks with fewer than two hundred nodes. 
\cite{random} relaxed maximin by allowing for randomized strategies, producing distributions over sets rather than a deterministic seed set, which differs from the study in this paper.
Recently, \cite{lp-greedy} revisited the maximin problem as a multiobjective submodular maximization task and proposed the LP-Greedy algorithm, which relies on the Gurobi Optimizer~\citep{gurobi}. 
While LP-Greedy achieves a ($1-1/e-\varepsilon$)-approximation under strict conditions (such as when the optimal solution is known or the solution of LP exceeds the optimal one), it remains computationally inefficient on large-scale networks.

In summary, despite the simplicity of the maximin constraint, existing studies face challenges in designing efficient algorithms with provable approximation guarantees under this constraint, hindering their practical applications to large-scale networks.
The fundamental difficulty lies in the fact that the maximin objective does not exhibit submodularity~\citep{groupfair}, a crucial property that enables efficient approximation with provable guarantees in traditional IM settings.

In this paper, we address this challenge by proposing efficient approximation algorithms for FIM under the maximin constraint, along with provable theoretical guarantees.
We introduce a two-step optimization framework: the first step performs inner-group maximization, where we show that influence spread within each group remains submodular though the maximin objective does not; and (2) the second step coordinates across-group maximization to improve the utility of the worst-off group with different coordination strategies, each offering provable approximation guarantees.

Our contributions can be summarized as follows:
\begin{itemize}
\setlength{\itemsep}{0pt}
\setlength{\parsep}{0pt}
    \item We propose a two-step framework to address the fair influence maximization problem under the maximin constraint. 
    The first step, Inner-group Maximization (IGM), applies the IMM algorithm to independently select seed nodes for each group, aiming to maximize inner-group influence spread.
    The second step, Across-group Maximization (AGM), coordinates seed selection across groups based on the seed lists output by IGM to improve the utility of the worst-off group.

    \item AGM introduces two seed selection strategies: Uniform Selection (AGM-US) and Greedy Selection (AGM-GS).
    AGM-US coordinates seed selection in a column-wise manner, establishing an approximation guarantee of roughly $\frac{1}{m}(1-1/e-\varepsilon)$ to the optimal solution regardless of group structure, with $m$ denoting the number of groups.
    AGM-GS employs a greedy manner, providing an $(1-1/e-\varepsilon)$ approximation when the groups are completely disconnected.
    
    \item Experimental results on seven real-world networks with different group structures demonstrate that our proposed framework consistently produces superior seed sets across all settings. 
    In particular, AGM-GS typically outperforms AGM-US, achieving better empirical performance than the proven theoretical bound of AGM-US while incurring a lower price of fairness. 
\end{itemize}
\section{Preliminaries}

A network can be modeled as a graph $G=(V,E)$, where $V$ represents the set of nodes and $E\subseteq V\times V$ denotes the set of edges that connect node pairs. 
$\cC=\{c_1, c_2, \cdots, c_m\}$ denotes the non-overlapping group (community) structure in $G$ with $V_i$ denote the set of nodes in $c_i$. 
Thus, it holds $V_i \cap V_j = \varnothing$ for any $c_i\neq c_j$.
An edge $e=(u, v)\in E$ with $u\in c_i, v\in c_j$ is called an inner-group edge if $c_i=c_j$.
On the contrary, the edge is called a cross-group edge if $c_i\neq c_j$.
In this paper, we define group connectivity $\rho(G,\cC)$ as the ratio between the number of cross-group edges and the number of total edges.
Formally, $\rho(G,\cC) = \frac{1}{|E|}\sum_{c\in\cC}\sum_{u\in V_c, v\notin V_c} \chi(u, v)$, where $\chi(u,v) = 1$ if $(u,v) \in E$, and $\chi(u,v) = 0$ otherwise.

\subsection{Independent Cascade Model}
The Independent Cascade (IC) model, originally proposed by~\cite{kempe}, is one of the most widely adopted information diffusion models in the field of Influence Maximization (IM).
The input of IC model includes a graph $G=(V,E)$ and a probability matrix $P\in [0,1]^{|V|*|V|}$, where each edge $e=(u,v)\in E$ is associated with a probability $p(e)=p(u,v)\in P$, representing the likelihood that node $u$ activates node $v$. 
In this study, we adopt a common setting for $P$, where $p(u,v)=1/d^-(v)$ with $d^-(v)$ represents the in-degree of node $v$~\citep{pmia,irie,imm,mic,FIMM-tkde}.

In the IC model, a node is either active or inactive. 
Starting with an initially active seed set $S$ at $t=0$, the model proceeds in discrete time steps. 
At each time step $t\ge 1$, nodes that were activated at time $t-1$ try to activate each of their inactive neighbors with a single, independent chance based on the corresponding edge probability. 
The set of nodes activated at time $t$ is denoted as $A_t$, with the initial set $A_0=S$. 
The whole diffusion ends at time step $t_\text{ed}$ when $A_{t_\text{ed}} = \varnothing$, which implies that no node is activated at $t_\text{ed}$. 

Note that our work supports the triggering model~\citep{kempe}, a general model that includes the IC model as a special case. 
We adopt IC solely for ease of comparison.

\subsection{Influence Spread}
The influence spread of a node set $S$, usually denoted as $\sigma(S)$, is a set function that represents the expected number of nodes activated. 
Formally, it is defined as $\sigma(S) = \mathbb{E}[|A_0 \cup A_1 \cup \cdots \cup A_{t_{\text{ed}}}|]$. 
In addition, we define $\sigma_c(S)$ as the inner influence spread for group $c$, where only activated nodes in $c$ are counted, namely, $\sigma_c(S) = \mathbb{E}[|A_0 \cup A_1 \cup \cdots \cup A_{t_{\text{ed}}}\cap V_c|]$.
The utility of group $c$ is then defined as $\ut_c(S) = \sigma_c(S) / |V_c|$, representing the influenced ratio within $c$.

\subsection{Live-edge Graph}
Given a graph $G = (V, E)$ and its associated probability matrix $P$, a live-edge graph $\cL$ can be sampled by independently selecting each edge $e = (u, v) \in E$ with its associated probability $p(u, v)\in P$.
Consequently, the probability of generating a particular live-edge graph $\cL$ is given by $\Pr\{\cL|G \} = \prod_{e\in E} p(e,\cL,G)$ where $p(e,\cL,G) = p(e)$ if $e\in \cL$, and $p(e,\cL,G) = 1-p(e)$ otherwise. 
Let $\Gamma(S, \cL)$ denote the set of nodes that can be reached in the live-edge graph $\cL$ by $S$. 
Then the influence spread $\sigma(S)$ can be expressed as $\sigma(S)=\sum_{\cL\in \Omega(\cL,G)} \Pr\{\cL|G\}\cdot|\Gamma(S, \cL)|$, where $\Omega(\cL,G)$ denotes the sample space of all possible live-edge graphs generated based on $G$.

\subsection{Maximin Constraint}
In this paper, we formulate the Fair Influence Maximization (FIM) problem under maximin constraints with maximin objective defined in the following Definition~\ref{def:MO}.
\begin{definition}\label{def:MO}
    (Maximin Objective) Given a graph $G=(V, E)$ with the group structure $\cC = \{c_1, c_2, \cdots, c_m \}$, the maximin constraint asks to find the seed set $S$ consisted of $k$ nodes that maximize the utility of the worst-off group, {\em i.e.}, to maximize $\Phi(S) = min_{c\in \cC}\ut_c(S)$.
    We assume the optimal seed set as $S^*$, such that $\Phi(S^*)=\max_{S\subset V, |S|=k} \Phi(S)$.
\end{definition}

\section{Method}

In this paper, we propose a two-step optimization framework to address the Fair Influence Maximization (FIM) problem under the maximin constraint.
The first step, Inner-group Maximization (IGM), employs a greedy approach to pick seed nodes for each group.
The second step, Across-group Maximization (AGM), coordinates seed selection across groups to regulate fairness at the global level.


\subsection{Inner-group Maximization (IGM)}

We first prove that $\sigma_c$ satisfies submodularity, which facilitates the use of a greedy approach to select seed nodes within each $c$ with provable approximation guarantees.

\begin{restatable}{fact}{factsub}\label{fact:sub}
(Submodularity) A set function $f: V\rightarrow \R$ is called submodular if for any $S \subseteq T \subseteq V$, $v\notin T$, $f$ holds
\begin{equation}
    f(S\cup\{v\}) - f(S) \geq f(T\cup \{v\}) - f(T).
\end{equation}
\end{restatable}

\begin{restatable}{lemma}{lemmagissub}\label{lem:gis_sub}
The inner influence spread $\sigma_c: S\rightarrow \R$ is submodular under the Independent Cascade (IC) model.
\end{restatable}

\begin{proof}
For a graph $G=(V,E)$ and its associated probability matrix $P$, let $\cL$ be a randomly sampled live-edge graph under the IC model and $\Gamma_c(S,\cL)$ denote the set of nodes in group $c$ that are reachable from $S$ within $\cL$.
Then, for any $S\subseteq T \subseteq V$, it holds $\Gamma_c(S,\cL) \subseteq \Gamma_c(T,\cL)$.
Let $v\in V, v\notin T$ be a random node, then it holds
\begin{align*}
    \Gamma_c(\{v\},\cL)\cap \Gamma_c(S,\cL) \subseteq \Gamma_c(\{v\},\cL)\cap \Gamma_c(T,\cL).
\end{align*}

Therefore, it establishes
\begin{align*}
    & |\Gamma_c(S\cup\{v\},\cL)| - |\Gamma_c(S,\cL)| \\
    = & |\Gamma_c(\{v\},\cL)| - |\Gamma_c(\{v\},\cL)\cap \Gamma_c(S,\cL)| \\
    \geq & |\Gamma_c(\{v\},\cL)| - |\Gamma_c(\{v\},\cL)\cap \Gamma_c(T,\cL)| \\
    = & |\Gamma_c(T\cup\{v\},\cL)| - |\Gamma_c(T,\cL)|.
\end{align*}

Recall that $\sigma(S)=\sum_{\cL\in \Omega(\cL,G)} \Pr\{\cL|G\}\cdot|\Gamma(S, \cL)|$, we have $\sigma_c(S)=\sum_{\cL\in \Omega(\cL,G)} \Pr\{\cL|G\}\cdot|\Gamma_c(S, \cL)|$, leading to
\begin{align*}
    & \sigma_c(S\cup \{v\})-\sigma_c(S) \\
    = & \sum_{\cL\in \Omega(\cL,G)} \Pr\{\cL|G\}\cdot \bigg( |\Gamma_c(S\cup\{v\}, \cL)| - |\Gamma_c(S, \cL)|\bigg) \\
    \geq & \sum_{\cL\in \Omega(\cL,G)} \Pr\{\cL|G\}\cdot \bigg( |\Gamma_c(T\cup\{v\}, \cL)| - |\Gamma_c(T, \cL)|\bigg) \\
    = & \sigma_c(T\cup \{v\})-\sigma_c(T),
\end{align*}
which concludes the proof.
\end{proof}

In addition, $\sigma_c$ is obviously non-negative and monotonic according to its definition.
Therefore, by applying a greedy approach that iteratively adds the node bringing the maximum gain {\em w.r.t} $\sigma_c$ to $S$, it will hold $\sigma_c(S)\ge(1-1/e)\cdot OPT$ \citep{kempe} with $OPT$ denoting the optimal result.

Thus, our proposed Inner-group Maximization (IGM) applies IMM~\citep{imm} within each group $c\in\cC$ to maximize their inner influence spread $\sigma_c(S^c)$ and obtain its correspoding seed set, denoted as $S^c = \{s^c_1, s^c_2, \cdots, s^c_k \}$.
Note that $s_i^c$ may lie outside group $c$.
Due to the fact that calculating $\sigma_c$ is \#P-hard~\citep{kempe}, IMM estimates $\sigma_c$ based on the Reverse Influence Sampling (RIS)~\citep{ris1,ris2} approach.
By generating $\theta_c=\frac{2n\cdot \left( (1-1/e)\cdot \sqrt{\ell \ln |V_c|+\ln4} +\sqrt{(1-1/e)\cdot \left( \ln\binom{|V_c|}{k}+\ell\ln{|V_c|}+\ln4 \right)}\right)^2}{OPT_c\cdot\varepsilon^2}$ Reverse Reachable (RR) sets~\citep{theta}, where $\varepsilon$ (usually 0.1) and $\ell$ (usually 1) are two manual parameters, the seed set $S^c$ identified via IMM guarantees the approximation of $\sigma_c(S^c)\geq(1-1/e-\varepsilon)\cdot OPT_c$ with a probability at least $1-1/|V_c|^\ell$.
Let $S^{*c}$ denote the optimal seed set which leads to $\sigma(S^{*c})=OPT_c$. 
It is essential to note that the seeds in $S^{*c}$ do not necessarily have to be inside $c$, which means it is a global solution.
The same formulation also holds for $S^c$.

In general, IGM outputs a set of inner-group seed sets $\cS=\{S^{c_1}, S^{c_2}, \cdots, S^{c_m} \}$, where each $S^c$ ($|S^c|=k$) corresponds to a group $c\in \cC$ and is produced using the IMM algorithm.
We denote the $i$-th seed node in $S^c$ as $s_i^c$, $i \in [k]$, as follows:
\begin{equation}\label{eq:igm_seeds}
\begin{split}
    & S^{c_1}:\;\ s^{c_1}_1,\ \ s^{c_1}_2,\ \ s^{c_1}_3,\ \ \cdots,\ \ s^{c_1}_k \\
    & S^{c_2}:\;\ s^{c_2}_1,\ \ s^{c_2}_2,\ \ s^{c_2}_3,\ \ \cdots,\ \ s^{c_2}_k \\
    & \cdots \\
    & S^{c_m}:\; s^{c_m}_1,\ s^{c_m}_2,\ s^{c_m}_3,\ \cdots,\ s^{c_m}_k
\end{split}
\end{equation}

Note that a node can appear in multiple lines, but no duplicates can exist within the same line.
In addition, when performing IMM in each group, we store those generated RR sets, which can be used to calculate $\ut_c(A), \forall c\in\cC, A\subset V$ with high efficiency in the next step. 

\subsection{Across-group Maximization (AGM)}
In this step, we propose two strategies for selecting the final seed set $S$ aimed at maximizing $\Phi(S)$: AGM with Uniform Selection (AGM-US), and AGM with Greedy Selection (AGM-GS).
Both strategies build upon the results produced by IGM, but differ in how they coordinate seed selection based on the outputs shown in Eq.~(\ref{eq:igm_seeds}).

Fundamentally, AGM-US leverages the inner-group submodularity to provide a general approximation bound applicable to arbitrary group structures.
In contrast, AGM-GS follows an intuitively greedy manner, though not bounded in most cases, it usually outperforms AGM-US in practice.

\subsubsection{Uniform Selection}
For AGM-US, seed nodes are selected in a column-wise manner.
Specifically, it selects $s^c_i$ only after $s^c_{i-1}$ has been selected as a seed for every $c\in\cC$. 
The complete procedure of AGM-US is outlined in Algorithm~\ref{alg:AGM-US}.

\begin{algorithm}[!h]
\small
\caption{AGM-US} 
\label{alg:AGM-US}
\begin{algorithmic}[1]
\Require 
Network $G$, group structure $\cC$, inner-group seed sets $\cS$, budget $k$, generated RR sets $\mathcal{R}$.
\Ensure 
MMF seed set $S$ with size $k$.
\State $S = \varnothing$
\State $k' =1$
\While {$|S|< k$}
    \State $S' = \bigcup_{c\in\cC} \{s^c_{k'}\} \setminus S $ 
    \If {$|S \cup S'|\leq k$}
        \State $S = S \bigcup S'$
        \State $k' = k'+1$
    \Else
        \For {$s\in S'$}
            \State $L(s)=\Phi(S\cup\{s \})$
        \EndFor
        \State $S = S\bigcup \{argmax_{s\in S'}L(s) \}$
    \EndIf
\EndWhile
\State Return $S$
\end{algorithmic}
\end{algorithm}

The idea of AGM-US is conceptually straightforward. 
It selects seeds from Eq.~(\ref{eq:igm_seeds}) column by column until it meets the budget.
The only exception occurs when it reaches the column (marked as the final column $k'$) where the remaining budget $k-|S|$ is less than the number of nodes not included in $S$ in that column.
In this case, AGM-US iteratively selects a node in the final column that could bring the maximum value to the objective $\Phi(S)$, until it meets the budget.
Note that AGM-US may often encounter duplicated nodes across different $S^c$, which leads to the final column usually exceeding $k/m$, where $m$ denotes the number of groups.
 
Let $k_c+q_c$ denote the number of nodes in $S$ that appear in $S^c$, where $k_c$ refers to the prefix of $S^c$ (thus $k_c\leq k'$) and $q_c$ ($q_c > k'$) denotes the remaining nodes.
Then, it is obvious that $|k_{c_i}-k_{c_j}|\leq 1, \forall c_i,c_j\in\cC, i\neq j$.
The following Theorem~\ref{thm:agmus} presents the approximation ratio of AGM-US with respect to the optimal solution.

\begin{restatable}{theorem}{theoremus}\label{thm:agmus}
    Let $S^*$ denote the optimal solution to Maximin Fairness, the seed set $S$ returned by AGM-US (Algorithm~\ref{alg:AGM-US}) satisfies $\Phi(S)\geq (\frac{1}{m}-\xi)(1-1/e-\varepsilon)\Phi(S^*)$ where $m$ is the number of communities and $\xi=\frac{mod(k,m)}{km}$.
\end{restatable}

\begin{proof}

    

    For any group $c$, let $S^c_r = \{s^c_1, s^c_2, \cdots, s^c_r \}$ ($r\leq k$) denoting the first $r$ seeds in $S^c$, and let $\Delta_c(s^c_i) = \sigma_c(S^c_i) - \sigma_c(S^c_{i-1})$ with $s^c_0 = \varnothing$. 
    By submodularity of $\sigma_c$ proved in Lemma~\ref{lem:gis_sub}, we have 
    \begin{equation*}
        \Delta_c(s^c_1) \geq \Delta_c(s^c_2) \geq \cdots \geq \Delta_c(s^c_k),
    \end{equation*}

    Therefore, for any $\sigma_c(S^c_r)$, $1\leq r \leq k$, it holds
    \begin{equation}\label{eq:sigmack}
    \begin{split}
        \sigma_c(S_r^c) 
        & = \Delta_c(s_1^c)+\Delta_c(s_2^c) + \cdots + \Delta_c(s_r^c) \\
        & = r \cdot \frac{1}{r}\sum_{i=1}^{r} \Delta_c(s_i^c) \\
        & \geq r \cdot\frac{1}{k}\sum_{i=1}^{k} \Delta_c(s_i^c) \\
        & = \frac{r}{k}\sigma_c(S_k^c).
    \end{split}
    \end{equation}

    Recall that the number of nodes selected in $S^c$ is $k_c + q_c$, which makes $k_c\geq \lfloor k/m \rfloor = (1/m - \xi)\cdot k$, where $\xi=\frac{mod(k,m)}{km}$.
    Then we have
    \begin{equation*}
    \begin{split}
        \sigma_c(S)  \geq \sigma_c(S^c_{k_c})
        & \geq (1-1/e-\varepsilon)\sigma_c(S^{*c}_{k_c})\ \ \text{(Given by IMM)}\\
        & \geq \frac{k_c}{k} (1-1/e-\varepsilon)\sigma_c(S^*)\ \ \text{(Based on Eq.~(\ref{eq:sigmack})}) \\
        & \geq (\frac{1}{m} - \xi)(1-1/e-\varepsilon)\sigma_c(S^*).
    \end{split}
    \end{equation*}

    Let $c^\#$ denote the group with the worst-off utility under seed set $S$ returned by AGM-US (Algorithm~\ref{alg:AGM-US}), such that $\Phi(S)=\ut_{c^\#}(S)$. 
    Similarly, let $c^*$ denote the group with the worst-off utility under the optimal seed set $S^*$, where $\Phi(S^*)=\ut_{c^*}(S^*)$.
    It finally holds
    \begin{equation*}
    \begin{split}
        \Phi(S) = \ut_{c^\#}(S) & \geq (\frac{1}{m} - \xi)(1-1/e-\varepsilon)\ut_{c^\#}(S^*) \\
        & \geq (\frac{1}{m} - \xi)(1-1/e-\varepsilon)\ut_{c^*}(S^*) \\
        & \geq (\frac{1}{m} - \xi)(1-1/e-\varepsilon)\Phi(S^*).
    \end{split}
    \end{equation*}

    Thus, we conclude the proof.
\end{proof}

As mentioned above, a node may exist duplicated across different $S^c$. 
Therefore, in practice, $k'$ is usually much larger than $\lfloor k/m \rfloor$, which leads to
\begin{equation*}
    \Phi(S) \geq \frac{k'}{k}(1-1/e-\varepsilon)\Phi(S^*),
\end{equation*}
providing a tighter and empirical bound.

\subsubsection{Greedy Selection}
For AGM-GS, seed nodes are selected in a greedy manner by comparing the current prefix node in each group to select the node that yields the maximum $\Phi$.
The prefix node of a group $c$ is initially set to its first seed, namely, $s^c_1$.
Once a seed is selected and added to $S$, each group updates its prefix node by advancing to the next node in its respective seed list until it finds a node not already included in $S$.
This update is necessary because the selected seed may appear duplicated in $S^c$ for different $c \in \cC$. 
Our update strategy ensures that the prefix node of each group is always the first node in its original order that remains available for selection, while still aiming to maximize its inner-group influence spread.
The detailed procedure of AGM-GS is depicted in Algorithm~\ref{alg:AGM-GS}.

\begin{algorithm}[!h]
\small
\caption{AGM-GS} 
\label{alg:AGM-GS}
\begin{algorithmic}[1]
\Require 
Network $G$, group structure $\cC$, inner-group seed sets $\cS$, budget $k$, generated RR sets $\mathcal{R}$.
\Ensure 
MMF seed set $S$ with size $k$.

\State $S = \varnothing$
\State set prefix indicator $k_c=1$, for all $c\in \cC$
\While {$|S|\leq k$}
    \State $S' = \bigcup_{c\in\cC} s^c_{k_c}$
    \State // Greedy selection:
    \For {$s\in S'$}
        \State $L(s)=\Phi(S\cup\{s\})$
    \EndFor 
    \State $S = S\bigcup \{argmax_{s\in S'}L(s) \}$
    \State // Indicator update:
    \For{$c\in \cC$}
        \While{$s^c_{k_c} \in S$}
            \State $k_c = k_c+1$
        \EndWhile
    \EndFor
\EndWhile
\State Return $S$
\end{algorithmic}
\end{algorithm}

The core idea of AGM-GS is intuitive: it continuously compares the current prefix nodes from each $S^c$ and greedily selects the one that gives the maximum $\Phi$. 
During each update, the prefix indicator of each group advances through its seed list until it meets a node that is not selected yet.
Unfortunately, though AGM-GS often outperforms AGM-US in practice, it does not offer any approximation guarantee to the optimal solution in general graph settings. 
However, if the groups of $\cC$ are completely disconnected in the graph $G$, {\em i.e.}, $\rho(G,\cC)=0$, then the seed set $S$ returned by AGM-GS satisfies the approximation bound $\Phi(S) \geq (1 - 1/e - \varepsilon)\Phi(S^*)$, as given by the following Theorem~\ref{thm:agmgs}.

\begin{restatable}{theorem}{theoremgs}\label{thm:agmgs}
    Let $S^*$ denote the optimal solution to Maximin Fairness, the seed set $S$ returned by AGM-GS (Algorithm~\ref{alg:AGM-GS}) satisfies $\Phi(S)\geq(1-1/e-\varepsilon)\Phi(S^*)$.
\end{restatable}

\begin{proof}
    Assume the number of seeds selected in each group $c$ from $S^*$ is $b_c$, {\em i.e.}, $|S^*\cap V_c|=b_c$, and the seed sets output by IGM is $\cS=\{S^{c_1}, S^{c_2}, \cdots, S^{c_m} \}$.
    Let $S^c_r = \{s^c_1, s^c_2, \cdots, s^c_r \}$ ($r\leq k$) denoting the first $r$ seeds in $S^c$ for each group $c$, $c\in\cC$, then we have
    \begin{equation}\label{eq:optimal-global2inner}
        \Phi(S^c_{b_c}) \geq (1-1/e-\varepsilon) \Phi(S^{*c}_{b_c})  = (1-1/e-\varepsilon) \Phi(S^*_{b_c}).
    \end{equation}
    
    Eq.~(\ref{eq:optimal-global2inner}) indicates that, when groups are completely disconnected, each group $c$ can only be influenced by its own seed set $S^c\subseteq V_c$.
    Consequently, selecting exactly the first $b_c$ nodes in $S^c$ as seeds for each $c \in \mathcal{C}$ yields a $(1 - 1/e - \varepsilon)$ approximation guarantee. 
    However, the actual number of nodes selected from $S^c$ by AGM-GS (Algorithm~\ref{thm:agmgs}), denoted as $d_c$, may not equal $b_c$ for all $c\in \cC$.
    
    Let $c^\#$ denote the group with the worst-off utility under seed set $S$ returned by AGM-GS (Algorithm~\ref{thm:agmgs}), such that $\Phi(S)=\ut_{c^\#}(S)$. 
    Similarly, let $c^*$ denote the group with the worst-off utility under the optimal seed set $S^*$, where $\Phi(S^*)=\ut_{c^*}(S^*)$. 
    Naturally, we consider two cases based on the number of selected nodes in $S^{c^\#}$: (1) $d_{c^\#} \geq b_{c^\#}$, and (2) $d_{c^\#} < b_{c^\#}$.
    




    (1) If $d_{c^\#}\geq b_{c^\#}$, we have
    \begin{equation*}
    \begin{split}
        \Phi(S) & = \ut_{c^\#}(S) = \ut_{c^\#}(S^{c^\#}_{d_{c^\#}}) \\
        & \geq \ut_{c^\#}(S^{c^\#}_{b_{c^\#}}) \geq (1-1/e-\varepsilon) \ut_{c^\#}(S^{*c^\#}_{b_{c^\#}}) \\
        & \geq (1-1/e-\varepsilon) \ut_{c^*}(S^{*c^*}_{b_{c^*}}) \\
        & = (1-1/e-\varepsilon) \Phi(S^*).
    \end{split}
    \end{equation*}

    (2) If $d_{c^\#}< b_{c^\#}$, since $\sum_{c\in\cC}d_c = \sum_{c\in\cC}b_c = k$, there will exist a group $c'$ which makes $d_{c'}>b_{c'}$, which leads to
    \begin{equation*}
    \begin{split}
        \ut_{c'}(S^{c'}_{d_{c'}-1}) 
        & \geq \ut_{c'}(S^{c'}_{b_{c'}}) \\
        & \geq (1-1/e-\varepsilon)\ut_{c'}(S^{*c'}_{b_{c'}}) \\
        & \geq (1-1/e-\varepsilon)\ut_{c^*}(S^{*c^*}_{b_{c^*}}) \\
        & = (1-1/e-\varepsilon)\Phi(S).
    \end{split}
    \end{equation*}

    Recall that Algorithm~\ref{alg:AGM-GS} applies a greedy strategy to coordinate seed selection.
    Thus, case (2) can only occur when     
    \begin{equation*}
    \begin{split}
        \Phi(S)  = \ut_{c^\#}(S) & =\ut_{c^\#}(S^{c^\#}_{d_{c^\#}})\\
        & \geq \ut_{c'}(S^{c'}_{d_{c'}-1}) \\
        & \geq (1-1/e-\varepsilon)\ut_{c'}(S^{*c'}_{b_{c'}}) \\
        & \geq (1-1/e-\varepsilon) \ut_{c^*}(S^{*c^*}_{b_{c^*}}) \\
        & \geq (1-1/e-\varepsilon)\Phi(S^*).
    \end{split}
    \end{equation*}
    
    Therefore, $\Phi(S) \geq (1-1/e-\varepsilon)\Phi(S^*)$ always holds for both cases, which concludes the proof.
\end{proof}

It should be noted that our proposed framework is most effective when the number of communities is small, which is our intended setting and also where fairness (e.g., gender or ethnicity) is most meaningful.
Besides, it also requires $k > m$, as this avoids the regime where the problem becomes inapproximable~\citep{krause2008robust}.

\subsection{Discussion on AGM and $\rho$}

The reason the maximin objective $\Phi(S)$ does not satisfy submodularity is that a group can also be influenced by nodes belonging to other groups. 
Recall that $\rho(G, \cC)$ measures the connectivity strength between groups, with $\rho \to 0$ indicating weaker connectivity and $\rho \to 1$ stronger connectivity.
When $\rho \to 0$, the likelihood that a group $c$ is influenced by nodes from other groups becomes negligible. 
In this case, AGM-GS is expected to perform well, approaching the approximation bound established in Theorem~\ref{thm:agmgs}.
Conversely, when $\rho \rightarrow 1$, across-group influence becomes more prevalent. 
This could result in significant overlaps among the seed sets $S^c \in \mathcal{S}$ for different $c\in \cC$, making $k'>>\lfloor k/m\rfloor$, which in turn enhances the approximation bound of AGM-US in Theorem~\ref{thm:agmus}.


\section{Experiments}

\subsection{Datasets}

We conduct experiments on seven real-world networks with ground-truth group structures.

(1) The AVC network~\citep{avc} models an obesity prevention intervention in the Antelope Valley region of California.
Each node in the network represents a real person associated with three demographic information.
In our experiment, we chose the gender feature to divide the network into Male (51\%) and Female (49\%).
(2) The labeled ENZ network~\citep{nr} models the relations between enzymes, and (3) the labeled PRO network~\citep{nr} models the relations between proteins.
Both the ENZ network and the PRO network are divided into three groups.
The group distributions for ENZ are 48.3\%, 49.4\%, and 2.3\%, and for PRO are 48.7\%, 48.1\%, and 3.2\%.
(4) The UVM network~\citep{uvmucsc} was built from the Facebook data in UVM (University of Vermont), where every node is a member of UVM. 
The network was preprocessed by~\cite{attributeRIS}, who remove the nodes without a user profile. 
It is divided by the Role feature into Faculty (12\%) and Student (88\%), or divided by the Grade feature into Senior (40\%) and Junior (60\%).
(5) Similarly, the UCSC dataset~\citep{uvmucsc} was built from the Facebook social network data in UCSC (University of California at Santa Cruz), where every node is a member of UCSC.
The network is either divided by the role feature into Faculty (10\%) and Student (90\%) or divided by the Gender feature into Male (45\%) and Female (55\%).
(6) The GOW network~\citep{gow} was collected from the public API of Gowalla, a location-based social networking service in which users share their locations via check-ins. 
Nodes are classified based on their check-in records into check-in users (54.5\%) and non-check-in users (45.5\%).
(7) The TWI network~\citep{twi} network was collected from the public API, where nodes represent Twitch users and edges denote mutual follower relationships. 
We partition the network based on the age feature into mature users (47.0\%) and immature users (53.0\%). 

Detailed statistics of those networks are given in Table~\ref{tab:dataset}.

\begin{table}[!h]
    \caption{Datasets.}
    \vspace{-5pt}
    \centering
    \small
    \begin{tabular}{rrrccc}
    \toprule
    Network & Nodes & Edges & $\rho$ & $m$ & Division \\
    \midrule
    AVC & 0.5k & 1.7k & 0.1889 & 2 & Gender \\
    ENZ & 19.6k & 74.6k & 0.3347 & 3 & Category \\
    PRO & 43.5k & 162.1k & 0.3432 & 3 & Category \\
    \midrule
    \multirow{2}{*}{UVM} & \multirow{2}{*}{7.3k} & \multirow{2}{*}{191.2k} & 0.1206 & 2 & Role \\
     &  &  & 0.2476 & 2 & Grade \\
    \midrule
    \multirow{2}{*}{UCSC} & \multirow{2}{*}{8.9k} & \multirow{2}{*}{224.5k} & 0.1093 & 2 & Role \\
    &  &  & 0.5025 & 2 & Gender \\
    \midrule
    GOW & 196.6k & 950.3k & 0.4187 & 2 & Check-info \\
    TWI & 168.1k & 6797.6k & 0.4547 & 2 & Age \\    
    \bottomrule
    \end{tabular}
    \label{tab:dataset}
\end{table}

\subsection{Baselines}

In experiments, we compare our proposed AGM-US and AGM-GS with 
(1) IMM~\citep{imm}: a traditional IM method that focuses on maximizing the influence spread without considering fairness; 
(2) Greedy~\citep{gaps}: a naive greedy approach that iteratively picks the node providing the maximum gain in the maximin objective;
and (3) LP-Greedy~\citep{lp-greedy}: an LP-based method that formulates the problem as multiobjective submodular maximization and solves it with Gurobi optimizer~\citep{gurobi}, has been shown to outperform previous approaches~\citep{multiobj,groupfair}.

It is worth noting that the primary time cost of our method arises from IMM. 
Our additional steps are lightweight: AGM-US runs in $O(k)$, while AGM-GS runs in $O(mk\theta)$.

\subsection{Maximin Objective}

\paragraph*{AVC.} 
We first evaluate our algorithm on the AVC network, which is partitioned by the gender attribute into a balanced group structure, with groups comprising $51\%$ and $49\%$ of the nodes, respectively.
Results are shown in Figure~\ref{fig:avc}.

\begin{figure}[!h]
    \centering
    \includegraphics[width=0.6\linewidth]{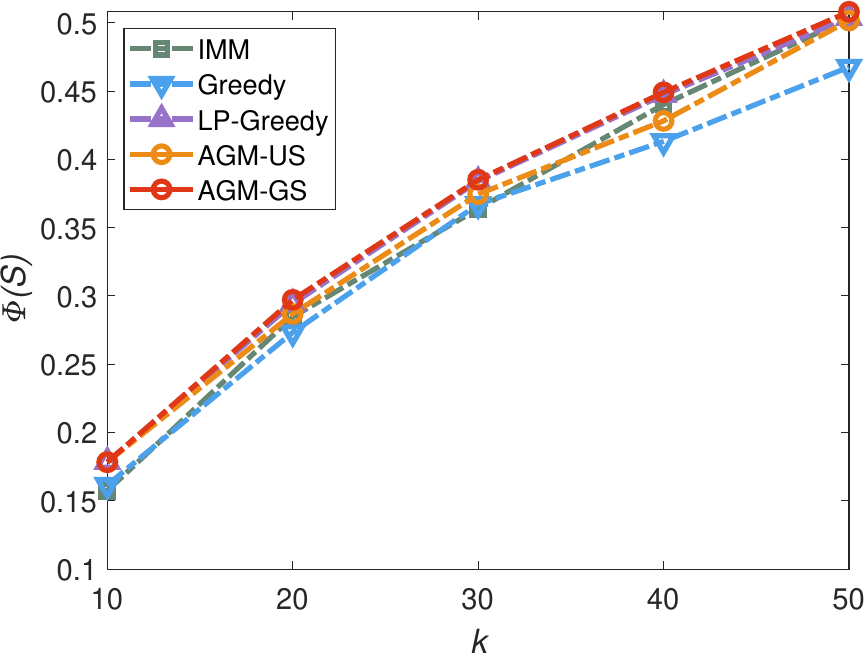}
    \vspace{-8pt}
    \caption{Comparison results on the AVC network.}
    \label{fig:avc}
\end{figure}

We consistently observe the performance ordering: AGM-GS $>$ LP-Greedy $>$ AGM-US $>$ Greedy, indicating that AGM-GS empirically surpasses the theoretical guarantee given by AGM-US.
In addition, IMM performs comparably to Greedy when $k$ is small, but its performance keeps approaching that of AGM-GS as $k$ increases.
This results from the balanced group structure, where maximizing the overall influence spread benefits to the worst-off group.

\paragraph*{ENZ \& PRO.}
Then, we conduct evaluations on ENZ and PRO, each partitioned into three groups.
Notably, both networks contain a minority group comprising approximately $3\%$ of the total nodes.
The results for ENZ and PRO are presented in Figure~\ref{fig:enzpro}(a) and Figure~\ref{fig:enzpro}(b), respectively.


\begin{figure}[!h]
  \centering
  \subfigure[ENZ]
  {\includegraphics[width=0.49\linewidth]{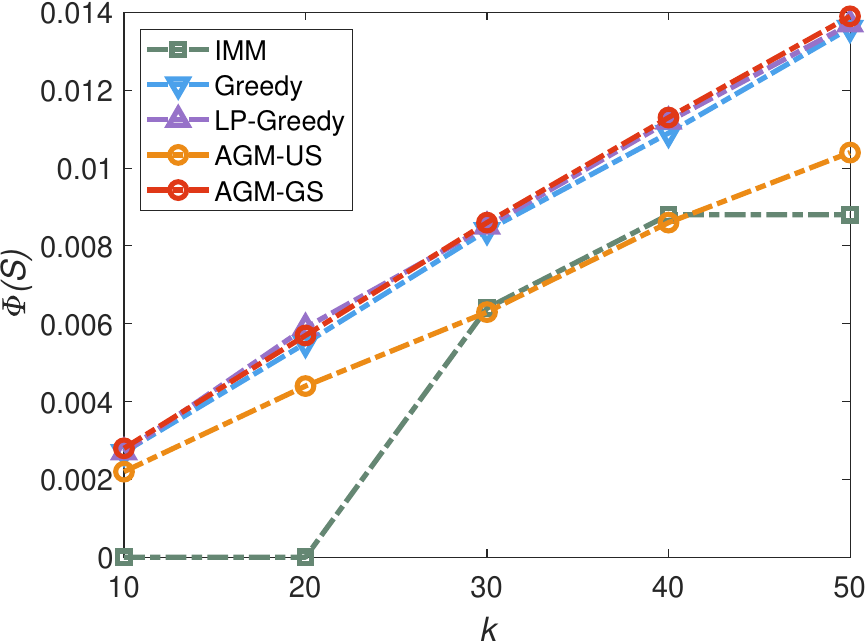}}
  \subfigure[PRO]
  {\includegraphics[width=0.49\linewidth]{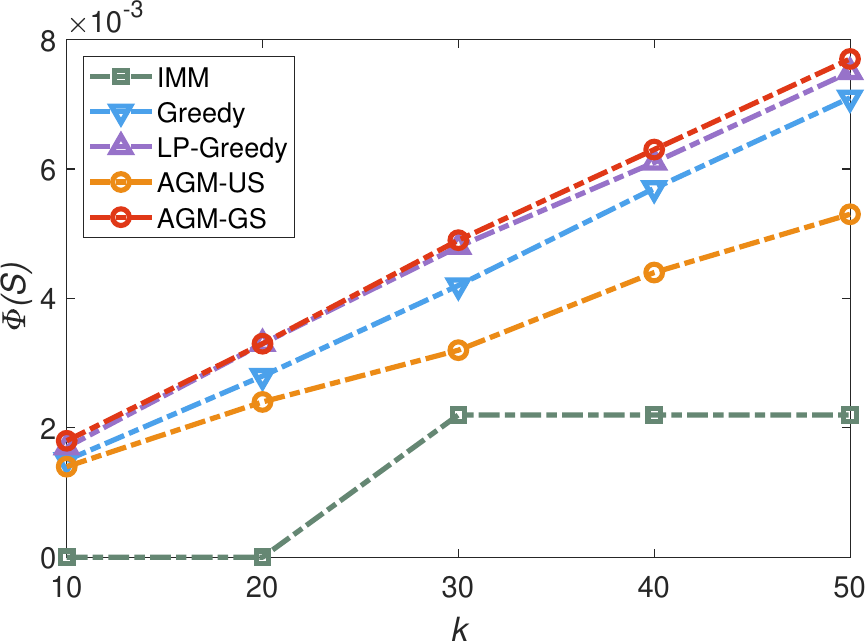}}
  \vspace{-8pt}
  \caption{Comparison results on the ENZ and PRO network.}
  \label{fig:enzpro}
\end{figure}

In contrast to the results on the AVC network, IMM performs poorly on both ENZ and PRO.
The reason lies in the fact that IMM focuses on maximizing global influence spread, which tends to overlook the minority group. 
Meanwhile, AGM-GS still consistently outperforms all other baselines, and notably surpasses LP-Greedy on the PRO network.

\begin{figure*}[!t]
  \centering
  \subfigure[UVM (Role)]
  {\includegraphics[width=0.245\linewidth]{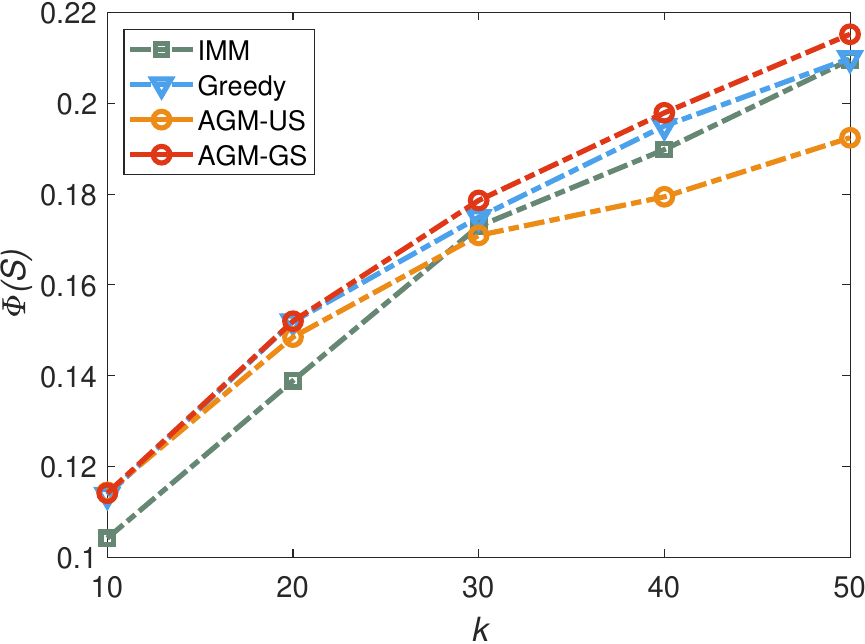}}
  \subfigure[UVM (Grade)]
  {\includegraphics[width=0.245\linewidth]{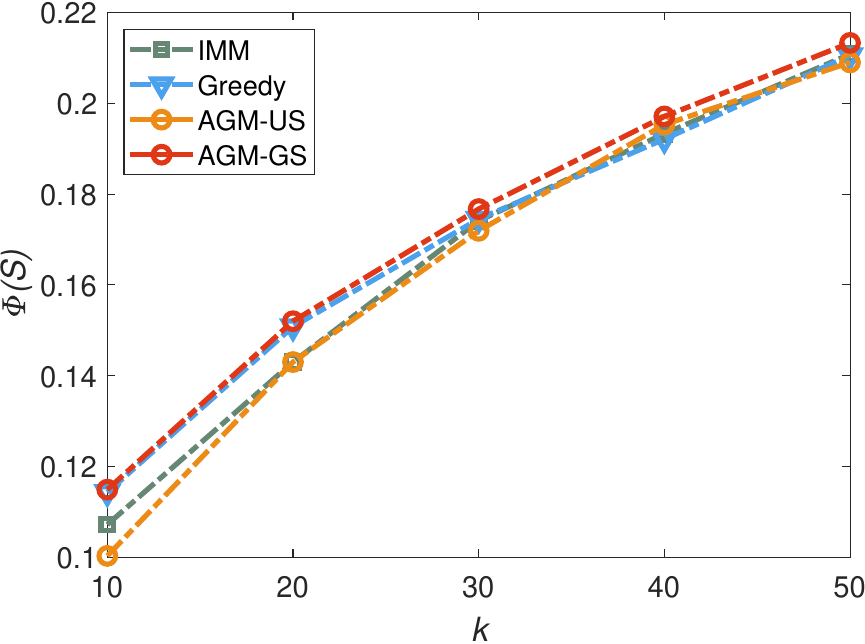}}
  \subfigure[UCSC (Role)]
  {\includegraphics[width=0.245\linewidth]{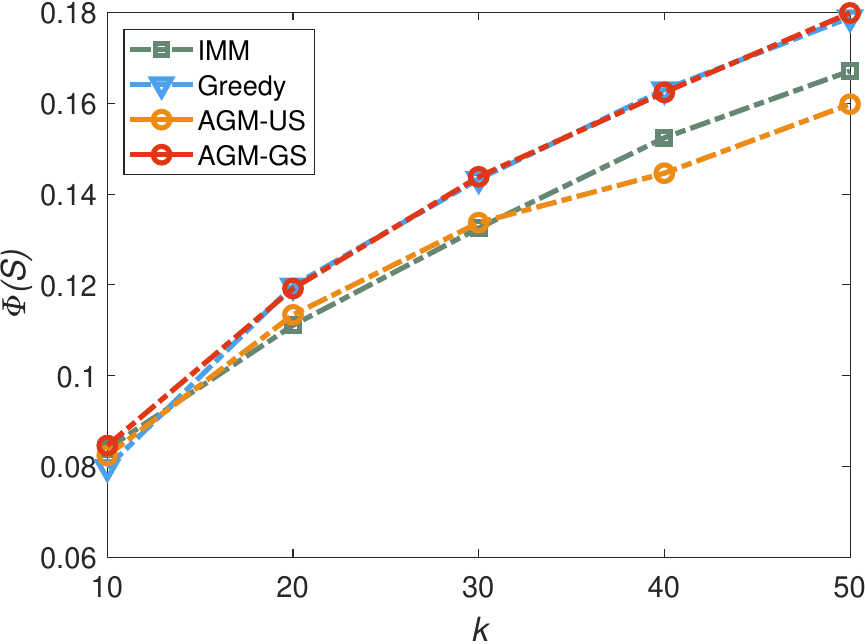}}
  \subfigure[UCSC (Gender)]
  {\includegraphics[width=0.245\linewidth]{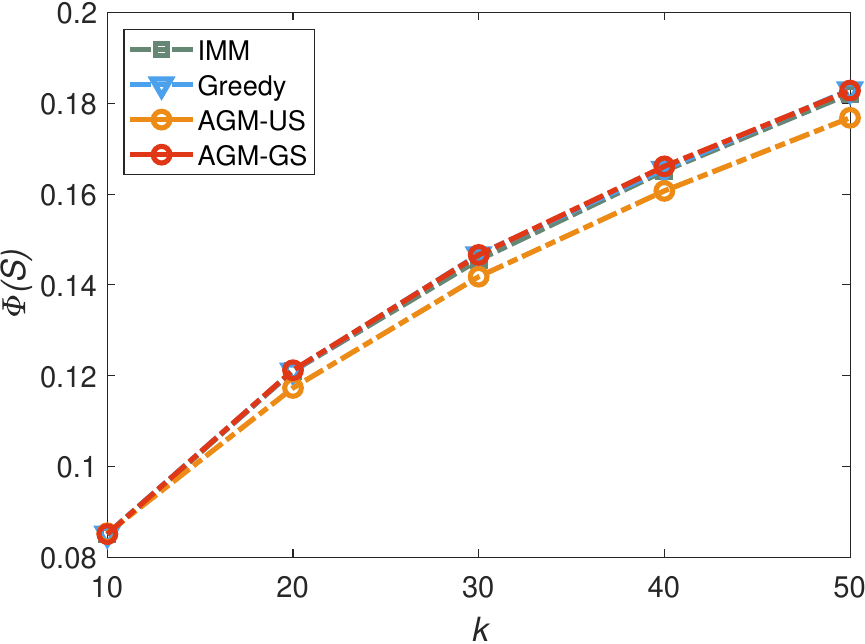}}
  \vspace{-8pt}
  \caption{Comparison results on the UVM and UCSC network.}
  \label{fig:uvmucsc}
\end{figure*}

\paragraph*{UVM \& UCSC.}
Both UVM and UCSC datasets feature two different group structures: one (based on "Grade" or "Gender") is balanced, with an approximate $50\%$:$50\%$ split, while the other (based on "Role") is highly imbalanced, with an approximate $10\%$:$90\%$ ratio.
Corresponding results are presented in Figure~\ref{fig:uvmucsc}, which includes four subfigures.
We attempted to run LP-Greedy on both UVM and UCSC. 
However, it failed to produce results after more than three days of computation, even with $k = 10$.

For UVM (Grade) and UCSC (Gender), where the group structures are balanced, all methods perform comparably.
In contrast, on UVM (Role) and UCSC (Role) with imbalanced group structures, it shows that AGM-GS $\geq$ Greedy $>$ IMM $>$ AGM-US.
Notably, IMM performs relatively well on these datasets due to the high network density, allowing the seeds of IMM to propagate effectively across different groups.

\paragraph*{GOW \& TWI.}

GOW and TWI are two large-scale networks, where TWI is even denser than UVM and UCSC.
For LP-Greedy, we obtained results only on GOW with $k = 10$, which cost 28 hours of computation. 

\begin{figure}[!h]
  \centering
  \subfigure[GOW]
  {\includegraphics[width=0.49\linewidth]{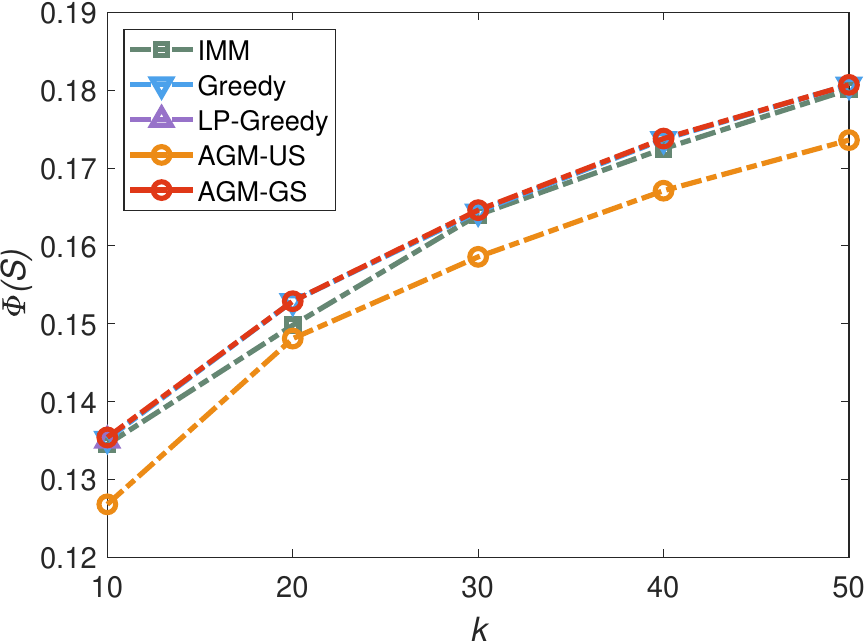}}
  \subfigure[TWI]
  {\includegraphics[width=0.49\linewidth]{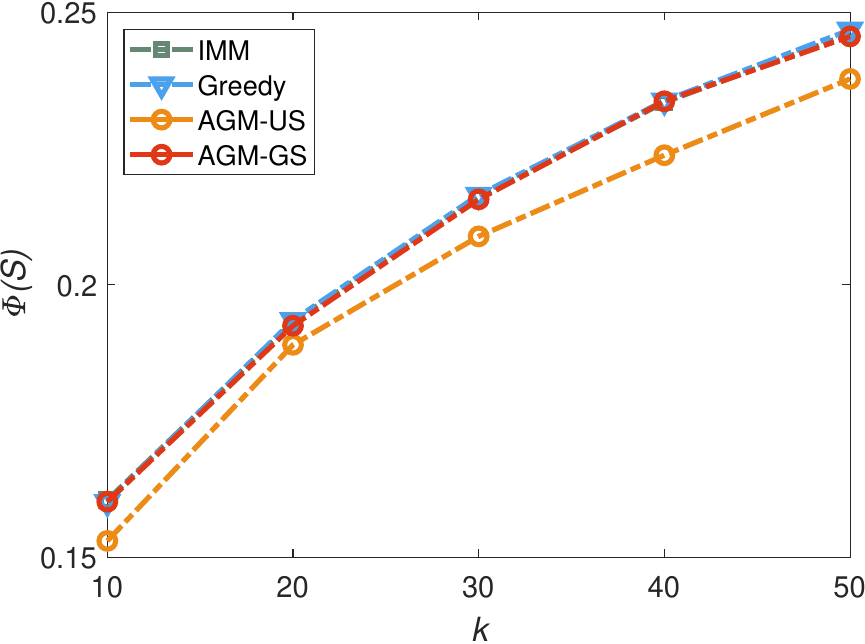}}
  \vspace{-8pt}
  \caption{Comparison results on the GOW and TWI network.}
  \label{fig:gowtwi}
\end{figure}

Results in Figure~\ref{fig:gowtwi} show that IMM, Greedy, and AGM-GS achieve comparable performance, with IMM being slightly worse.
This phenomenon is consistent with the results on UVM (Grade) and UCSC (Gender), where the group structures are relatively balanced (approximately $45\%$:$55\%$ for both GOW and TWI).
Notably, although Greedy attains similar performance, it requires 14 hours to produce results on TWI, whereas AGM-GS requires only approximately 1.1 times the runtime of IMM, completing in tens of minutes.

\vspace{-5pt}
\subsection{Price of Fairness and Runtime}
Price of Fairness (PoF) measures the trade-off between fairness and overall influence spread, quantifying how much influence is sacrificed to achieve fairness.
Specifically, we use the influence spread achieved by IMM as the baseline and compute the PoF of any seed set $S$ as $PoF(S) = \frac{\sigma(S^\text{IMM})-\sigma(S)}{\sigma(S^\text{IMM})}$ with $S^\text{IMM}$ denoting the seed set of IMM.
The PoF results across different networks are summarized in Table~\ref{tab:pof}, which also lists the runtime results for $k=50$.

Generally, our proposed AGM-GS consistently achieves the best performance, incurring the lowest influence loss while attaining higher fairness.
On UCSC (Gender), GOW, and TWI, the PoF of AGM-GS is comparable to Greedy, likely due to their higher group connectivity $\rho$. 
As discussed earlier, AGM-GS is most effective as $\rho \rightarrow 0$;
Hence, its performance slightly deteriorates in scenarios with higher $\rho$.

Moreover, both AGM-US and AGM-GS exhibit excellent scalability, with average runtimes less than one-tenth of that of Greedy. 
In contrast, LP-Greedy incurs substantially higher computational cost than Greedy and fails to produce results on dense or large-scale networks. 

\begin{table}[!t]
\centering
\setlength{\tabcolsep}{1.5mm}
\caption{PoF for $k=10,30,50$ and runtime (s) for $k=50$. \textbf{Bold} indicates the best PoF, and \sed{underlined italic} denotes the runner-up.}
\vspace{-5pt}
\small
\begin{tabular}{clrrrrr}
\toprule
\textbf{Network} & \textbf{Method} & \textbf{$k=10$} & \textbf{$k=30$} & \textbf{$k=50$} & runtime \\
\midrule
\multirow{4}{*}{AVC}
& Greedy    & 13.20\%  & 6.26\%  & 9.83\%  & 4.1 \\
& LP-Greedy   & \textbf{2.14\%}  & \textbf{1.65\%} & \sed{2.93\%} & 4285.8 \\
& AGM-US   & \textit{\underline{4.11\%}}  & 2.43\%  & 2.95\%  & 0.5 \\
& AGM-GS   & 4.24\%  & \textit{\underline{2.07\%}}  & \textbf{2.38\%} & 1.3 \\
\midrule
\multirow{4}{*}{ENZ}
& Greedy    & 23.78\% & 14.72\% & 13.58\% & 282.7\\
& LP-Greedy   & \sed{17.26\%} & \sed{13.64\%}  & \sed{12.64\%} & 1711.9 \\
& AGM-US   & 19.77\% & 16.42\%  & 14.00\%  & 29.3 \\
& AGM-GS   & \textbf{14.42\%} & \textbf{11.59\% } & \textbf{11.38\%} & 78.5 \\
\midrule
\multirow{4}{*}{PRO}
& Greedy    & 35.16\%  & 28.39\%  & 22.12\% & 616.7 \\
& LP-Greedy   & 25.77\%  & \sed{18.91}\%  & 16.67\% & 2355.4 \\
& AGM-US   & \sed{24.93\%}  & 20.39\%  & \sed{16.54\%} & 43.9 \\
& AGM-GS   & \textbf{21.88\%} & \textbf{18.00\% } & \textbf{14.68\%} & 114.1 \\
\midrule
\multirow{3}{*}{\shortstack{UVM\\ \\(Role)}}
& Greedy    & 2.49\%  & \sed{2.99\%}  & \sed{2.78\%} & 1008.6  \\
& AGM-US   & \textbf{1.73\%}  & 3.65\%  & 7.41\% & 69.1 \\
& AGM-GS   & \sed{1.90\%}  & \textbf{0.78\%}  & \textbf{0.29\%} & 75.0  \\
\midrule
\multirow{3}{*}{\shortstack{UVM\\ \\(Grade)}}
& Greedy    & 0.54\%  & 1.72\%  & \sed{1.82\%}  & 1098.6 \\
& AGM-US   & \sed{0.49\%}  & \textbf{0.00\%}  & 1.89\%  & 69.8 \\
& AGM-GS   & \textbf{0.18\%}  & \sed{0.74\%}  & \textbf{1.28\%}  & 76.3 \\
\midrule
\multirow{3}{*}{ \shortstack{UCSC\\ \\(Role)}}
& Greedy    & 6.52\%  & \sed{1.71\%}  & \textbf{1.57\%} & 1269.5 \\
& AGM-US   & \sed{3.17\%}  & 6.34\%  & 9.77\% & 68.5  \\
& AGM-GS   & \textbf{-0.35\%}  & \textbf{0.93\%}  & \sed{1.72\%} & 74.3 \\
\midrule
\multirow{3}{*}{\shortstack{UCSC\\ \\(Gender)}}
& Greedy    & \sed{0.13\%} & \textbf{0.14\%}  & \textbf{0.00\%} & 1203.3\\
& AGM-US   & \textbf{0.01\%} & 1.99\%  & 2.14\% & 102.8 \\
& AGM-GS   & 0.25\%  & \sed{0.21\%}  & \sed{0.26\%} & 112.7 \\
\midrule
\multirow{4}{*}{GOW}
& Greedy    & \textbf{0.98\%} & \sed{0.83\%} & \sed{1.24\%} & 2920.5\\
& LP-Greedy   & 1.39\% & -  & - & - \\
& AGM-US   & 7.00\% & 3.93\%  & 4.36\% & 157.1 \\
& AGM-GS   & \sed{1.15\%} & \textbf{0.76\% } & \textbf{1.19\%} & 172.2 \\
\midrule
\multirow{3}{*}{TWI}
& Greedy    & \sed{0.19\%} & \textbf{0.22\%} & \textbf{0.29\%} & 49905.3 \\
& AGM-US   & 4.70\%  & 3.73\%  & 3.56\% & 2663.2 \\
& AGM-GS   & \textbf{0.13\%} & \sed{0.33\% } & \sed{0.39\%} & 2682.1 \\
\bottomrule
\end{tabular}
\label{tab:pof}
\end{table}

\section{Conclusion}


This paper studies the Fair Influence Maximization (FIM) problem under the maximin objective, which aims to maximize the utility of the worst-off group. 
To address the non-submodularity of the maximin objective, we propose a two-step optimization framework that leverages the submodularity of influence spread within individual groups, as formally established. 
The first step, Inner-group Maximization (IGM), generates high-quality group-wise seed candidates with theoretical guarantees, while the second step, Across-group Maximization (AGM), coordinates seed selection across groups. 
AGM also includes two variants: AGM-US, which provides a $\frac{1}{m}(1 - 1/e - \varepsilon)$ approximation independent of group structure, and AGM-GS, which achieves a $(1 - 1/e - \varepsilon)$ approximation when groups are disconnected. 
Extensive experiments on seven real-world networks demonstrate that our methods consistently outperform state-of-the-art baselines, with AGM-GS typically achieving superior empirical performance.

\bibliographystyle{named}
\bibliography{ijcai26}

\end{document}